\documentclass[a4paper]{amsproc}
\usepackage{amssymb}
\usepackage{amscd}

\usepackage{pstricks}
\usepackage[dvips]{graphicx}
\usepackage{pst-3dplot}
\usepackage[all,cmtip]{xy}

\usepackage{pst-solides3d}
\usepackage{pstricks-add}
\usepackage{amsfonts}
\usepackage{cite}
\usepackage{amsmath}
\usepackage{hyperref}

\theoremstyle{plain}
 \newtheorem{thm}{Theorem}[section]
 \newtheorem{prop}{Proposition}[section]
 \newtheorem{lem}{Lemma}[section]
 \newtheorem{cor}{Corollary}[section]
\theoremstyle{definition}
 
 \newtheorem{remark}{Remark}[section]

\numberwithin{equation}{section}

\newcommand{\R}{\mathbb{ R}}

\newcommand{\pr}{\mathrm{pr}}

\newcommand{\diag}{\mathrm{diag}}

\newcommand{\tr}{\mathrm{tr}}

\newcommand{\dv}{\mathrm{ div}}

\title[Spherical and planar ball bearings]{Spherical and planar ball bearings -- nonholonomic systems with invariant measures}

\author[V. Dragovi\'c, B. Gaji\'c, B. Jovanovi\'c]{\bfseries Vladimir Dragovi\'c, Borislav Gaji\'c, Bo\v zidar Jovanovi\'c}

\dedicatory{Dedicated to the memory of Professor Alexey Vladimirovich Borisov}

\address{
Department of Mathematical Sciences  \\
The University of Texas at Dallas   \\
Richardson, TX\\
USA\\
Mathematical Institute\\
Serbian Academy of Sciences and Arts\\
Belgrade\\
Serbia}
\email{Vladimir.Dragovic@utdallas.edu}

\address{
Mathematical Institute - the corresponding author\\
Serbian Academy of Sciences and Arts\\
Belgrade\\
Serbia}
\email{gajab@mi.sanu.ac.rs}

\address{
Mathematical Institute\\
Serbian Academy of Sciences and Arts\\
Belgrade\\
Serbia}
\email{bozaj@mi.sanu.ac.rs}

\subjclass[2010]{37J60, 37J35, 70E40,  70F25}

\keywords{Nonholonimic dynamics; rolling without slipping, invariant measure; integrability}

\begin{document}

\begin{abstract}
We first construct nonholonomic systems of $n$ homogeneous balls $\mathbf B_1,\dots,\mathbf B_n$ with centers $O_1,...,O_n$ and with the same radius $r$ that are rolling without slipping around a fixed sphere $\mathbf S_0$ with center $O$ and radius $R$.
In addition, it is assumed that a dynamically nonsymmetric sphere $\mathbf S$ of radius $R+2r$ and the center that coincides with the center $O$ of the fixed sphere  $\mathbf S_0$ rolls without slipping over the moving balls $\mathbf B_1,\dots,\mathbf B_n$.
We prove that these systems possess an invariant measure.
As the second task, we consider the limit, when the radius $R$ tends to infinity. We obtain a corresponding planar problem consisting of
$n$ homogeneous balls $\mathbf B_1,\dots,\mathbf B_n$ with centers $O_1,...,O_n$ and the same radius $r$ that are rolling without slipping
over a fixed plane $\Sigma_0$, and a moving plane $\Sigma$  that moves without slipping over the homogeneous balls.
We prove that this system possesses an invariant measure and that it is integrable in quadratures according to the Euler-Jacobi theorem.

\end{abstract}

\maketitle

\section{Introduction}

In this paper, we first construct nonholonomic systems of $n$ homogeneous balls $\mathbf B_1,\dots,\mathbf B_n$ with centers $O_1,...,O_n$ and with the same radius $r$ that are rolling without slipping around a fixed sphere $\mathbf S_0$ with center $O$ and radius $R$.
We assume that a dynamically nonsymmetric sphere $\mathbf S$ of radius $R+2r$ and the center that coincides with the center $O$ of the fixed sphere  $\mathbf S_0$ rolls without slipping over the moving balls $\mathbf B_1,\dots,\mathbf B_n$.
 The rolling of the balls $\mathbf B_i$ and the sphere $\mathbf S$ are considered under the inertia and in the absence of external forces.
 We refer to this system as a \emph{spherical ball bearing}  (see Figure 1).

As the second task, we consider the limit, when the radius $R$ tends to infinity. In that way, we obtain a corresponding planar problem consisting of
$n$ homogeneous balls $\mathbf B_1,\dots,\mathbf B_n$ with centers $O_1,...,O_n$ and the same radius $r$ that are rolling without slipping
over a fixed plane $\Sigma_0$, and a moving plane $\Sigma$
{that} moves without slipping over the homogeneous balls.
We refer to this system as a \emph{planar ball bearing}  (see Figure 2).

Although the rolling ball problems are very well studied (see \cite{BM2002, BMB2013, BMT2014, BMbook}), the spherical and planar bearing problems seem not to be considered before. There are two nonholonomic systems which are close to the spherical ball bearings. One is the so-called spherical support system, introduced by Fedorov in \cite{F1}. It describes the rolling without slipping of a dynamically nonsymmetric sphere $\mathbf S$ over $n$ homogeneous balls $\mathbf B_1, \dots,\mathbf B_n$ of possibly different radii, but with fixed centers. The second one is the rolling of a homogeneous ball $\mathbf B$ over a dynamically asymmetric sphere $\mathbf S$, introduced by
Borisov, Kilin, and Mamaev in \cite{BKM}. They considered in \cite{BKM} both situations: when the center of $\mathbf S$ is fixed, and when it is not.

In Section \ref{sec2} we define the spherical ball bearing system: the configuration space $Q$, the nonholonomic distribution $\mathcal D\subset TQ$ and the Lagrangian that coincides with the kinetic energy of the system. The kinetic energy and the distribution  are invariant with respect to an appropriate action of the Lie group $SO(3)^{n+1}$,
and the system can be reduced to $\mathcal M=\mathcal D/SO(3)^{n+1}$.
In Section \ref{sec3} we derive the equations of motion of the reduced spherical ball bearing system in terms of the reaction forces and list some first integrals in Propositions \ref{prva} and \ref{go}.  Proposition \ref{prva} implies that the centers $O_1,...,O_n$ are in rest in relation to each other. Thus, there are no collisions of the balls $\mathbf B_1,\dots,\mathbf B_n$.

\

\begin{pspicture}(12,7)
\pscircle[linecolor=black,fillstyle=solid, fillcolor=gray!20, linestyle=dashed](3.5,3){0.88}
\psellipticarc[linestyle=dashed](3.5,3)(0.88,0.2){-60}{180}
\psellipticarc(3.5,3)(0.88,0.2){180}{300}
\psdot[dotsize=2pt](3.5,3)\uput[0](2.7,3.2){$O_1$}
\psellipticarc[linecolor=black](3.5,3)(0.895,0.895){102}{304}
\psdot[dotsize=2pt](4,3.27)\uput[0](4,3.27){$A_1$}
\psdot[dotsize=2pt](3,2.73)\uput[0](3,2.73){$B_1$}

\pscircle[linecolor=black, fillstyle=solid, fillcolor=gray!20](7,3){1}
\psellipticarc[linestyle=dashed](7,3)(1,0.2){0}{180}
\psellipticarc(7,3)(1,0.2){180}{360}
\psdot[dotsize=2pt](7,3)\uput[0](7,3.2){$O_2$}
\psdot[dotsize=2pt](7.8,2.5)\uput[0](7.8,2.5){$B_2$}
\psdot[dotsize=2pt](6.2,3.5)\uput[0](6.2,3.5){$A_2$}

\pscircle[linecolor=black, fillstyle=solid, fillcolor=gray!20, linestyle=dashed](5.5,5){0.8}
\psellipticarc[linestyle=dashed](5.5,5)(.8,0.2){0}{180}
\psellipticarc[linestyle=dashed](5.5,5)(0.8,0.2){180}{360}
\psdot[dotsize=2pt](5.5,5)\uput[0](5.5,5.2){$O_3$}
\psdot[dotsize=2pt](5.38,4.4)\uput[0](5.38,4.6){$A_3$}
\psdot[dotsize=2pt](5.62,5.6)\uput[0](5.52,5.8){$B_3$}

\pscircle[linecolor=black, linestyle=dashed](5.3,4){2}
\pscircle[linecolor=black](5.3,4){3}
\psellipticarc[linecolor=black](5.3,4)(2.015,2.015){0}{300}
\psdot[dotsize=2pt](5.3,4)\uput[0](5.3,4){$O$}
\psellipticarc[linestyle=dashed](5.3,4)(2,0.2){-10}{180}
\psellipticarc[](5.3,4)(2,0.2){180}{350}

\uput[0](3.3,4.5){$\mathbf S_0$}
\uput[0](2.6,5.5){$\mathbf S$}
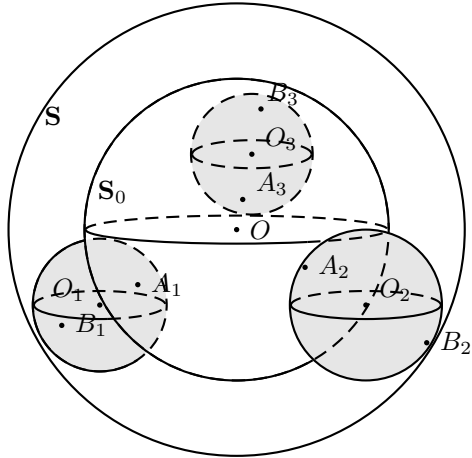
\uput[0](2,0.5){{\sc Figure 1}. Spherical ball bearing for $n=3$}
\end{pspicture}

In Section \ref{sec4} we perform the second reduction by fixing the values of the $n$ first
integrals from Proposition \ref{go}. We obtain the closed system of equations of motion of the reduced spherical ball bearing system on the space $\mathcal N=\R^3\times (S^2)^n$ in Theorem \ref{Glavna}.
The complete set of non-reduced equations of motion is given in Corollary \ref{Originalne}.
Finally, we prove that the spherical ball bearing problem has an invariant measure in Theorem \ref{mera}, in Section \ref{sec5}. The question of integrability in  the spherical ball bearing problem will be studied in a separate paper.

In Section \ref{sec6} we consider the planar ball bearing problem. To simplify notation, we consider the case $n=3$ and refer to the system as the
\emph{three balls planar bearing problem}.
However all the  statements and considerations from Section \ref{sec6} hold for arbitrary $n$ in a straightforward manner.

For general $n$, the configuration space and the nonholonomic distribution $\mathcal D$ are of the same dimensions
as in the spherical ball bearing problem but the description of the system is slightly different.
For $n=3$, we derive the equations of motion on $\mathcal D/SO(3)\times SO(3)\times SO(3)$, see Theorem \ref{Glavna2}. We perform a second reduction to a space
$\mathcal Q\subset \R^6$ defined by an algebraic inequality. We prove that the planar three balls bearing problem on $\mathcal Q$ has an invariant measure and
four independent first integrals. Therefore, it is integrable according to the Euler-Jacobi theorem, see Theorem \ref{mera2}.

\section{Rolling of a dynamically nonsymmetric sphere over $n$ moving homogeneous balls and a fixed sphere}\label{sec2}

We consider the following spherical ball bearing problem:
$n$ homogeneous balls $\mathbf B_1,\dots,\mathbf B_n$ with centers $O_1,...,O_n$ and the same radius $r$  roll without slipping around a fixed sphere $\mathbf S_0$ with center $O$ and radius $R$. A dynamically nonsymmetric sphere $\mathbf S$ of radius $R+2r$ with the center that coincides with the center $O$ of the fixed sphere $\mathbf S_0$  rolls without slipping over the moving balls $\mathbf B_1,\dots,\mathbf B_n$.

For $n \ge 4$ there are initial positions of the balls $\mathbf B_1,\dots,\mathbf B_n$ that imply the condition that the centre of the moving sphere $\mathbf S$ coincides with the centre $O$ of the fixed sphere $\mathbf S_0$. Let us reiterate that the configuration of the balls is congruent during the time evolution.
In order to include all possible initial positions for arbitrary $n$, the condition that $O$ coincides with the centre of the sphere $\mathbf S$ is assumed to be a holonomic constraint.

Let
\[
O\vec{\mathbf e}^0_1,\vec{\mathbf e}^0_2,\vec{\mathbf e}^0_3, \qquad O\vec{\mathbf e}_1,\vec{\mathbf e}_2,\vec{\mathbf e}_3,
\qquad O_i\vec{\mathbf e}^i_1,\vec{\mathbf e}^i_2,\vec{\mathbf e}^i_3, \qquad i=1,\dots,n
\]
be positively oriented reference frames rigidly attached to the spheres $\mathbf S_0$, $\mathbf S$, and the balls $\mathbf B_i$, $i=1,\dots,n$, respectively.
By $\mathbf g,\mathbf g_i\in SO(3)$ we denote the matrices that map the
moving frames $O\vec{\mathbf e}_1,\vec{\mathbf e}_2,\vec{\mathbf e}_3$ and $O_i\vec{\mathbf e}^i_1,\vec{\mathbf e}^i_2,\vec{\mathbf e}^i_3$
to the fixed frame $O\vec{\mathbf e}^0_1,\vec{\mathbf e}^0_2,\vec{\mathbf e}^0_3$:
\[
\mathbf g_{jk}=\langle \vec{\mathbf e}^0_j,\vec{\mathbf e}_k \rangle, \qquad  \mathbf g_{i,jk}=\langle \vec{\mathbf e}^0_j,\vec{\mathbf e}^i_k \rangle,
\qquad  j,k=1,2,3, \qquad i=1,\dots,n.
\]

We apply the standard isomorphism between the Lie algebras $(so(3),[\cdot,\cdot])$ and $(\R^3,\times)$
given by
\begin{equation}\label{izomorfizam}
a_{ij}=-\varepsilon_{ijk}a_k, \qquad i,j,k=1,2,3,
\end{equation}
The skew-symmetric matrices
\[
\omega=\dot{\mathbf g}{\mathbf g}^{-1}, \qquad \omega_i=\dot{\mathbf g}_i \mathbf g_i^{-1}
\]
correspond to the angular velocities $\vec{\omega}$, $\vec{\omega}_i$ of the sphere $\mathbf S$ and the $i$-th ball $\mathbf B_i$ in the fixed reference frame
$O\vec{\mathbf e}^0_1,\vec{\mathbf e}^0_2,\vec{\mathbf e}^0_3$ attached to the sphere $\mathbf S_0$. The matrices
\[
\Omega=\mathbf g^{-1}\dot{\mathbf g}=\mathbf g^{-1}\omega\mathbf g, \qquad W_i=\mathbf g^{-1}_i\dot{\mathbf g}_i=\mathbf g^{-1}_i\omega_i\mathbf g_i
\]
correspond to the angular velocities $\vec{\Omega}$, $\vec{W}_i$ of $\mathbf S$ and $\mathbf B_i$ in the frames $O\vec{\mathbf e}_1,\vec{\mathbf e}_2,\vec{\mathbf e}_3$
and  $O_i\vec{\mathbf e}^i_1,\vec{\mathbf e}^i_2,\vec{\mathbf e}^i_3$ attached to the sphere $\mathbf S$ and the balls $\mathbf B_i$, respectively.

We have
\[
\vec{\omega}=\mathbf g\vec{\Omega}, \qquad \vec{\omega}_i=\mathbf g_i \vec{W}_i.
\]

Let $I$ be the inertia operator of the outer sphere $\mathbf S$.
We choose the moving frame $O\vec{\mathbf e}_1,\vec{\mathbf e}_2,\vec{\mathbf e}_3$,
such that $O\vec{\mathbf e}_1$, $O\vec{\mathbf e}_2$, $O\vec{\mathbf e}_3$ are the principal axes of inertia: $I=\diag(A,B,C)$.
Let $\diag(I_i,I_i,I_i)$ and $m_i$  be the inertia operator and the mass of the $i$-th ball $\mathbf B_i$.
Then the configuration space and the kinetic energy  of the problem are given by:
\begin{align*}
Q=&SO(3)^{n+1}\times (S^2)^{n}\{\mathbf g,\mathbf g_1,\dots,\mathbf g_n,\vec\gamma_1,\dots,\vec\gamma_n\},\\
T=&\frac12 \langle I\vec\Omega,\vec\Omega\rangle+\frac12 \sum_{i=1}^n I_i \langle \vec{W}_i,\vec{W}_i\rangle+\frac12 \sum_{i=1}^n m_i \langle \vec{v}_{O_i},\vec{v}_{O_i}\rangle\\
=&\frac12 \langle I\vec\Omega,\vec\Omega\rangle+\frac12 \sum_{i=1}^n I_i \langle \vec{\omega}_i,\vec{\omega}_i\rangle+\frac12 \sum_{i=1}^n m_i \langle \vec{v}_{O_i},\vec{v}_{O_i}\rangle.
\end{align*}
Here $\vec{\gamma}_i$ is the unit vector
\[
\vec{\gamma}_i=\frac{\overrightarrow{OO_i}}{|\overrightarrow{OO_i}|}
\]
determining the position $O_i$ of the centre of $i$-th ball $\mathbf B_i$ and $\vec{v}_{O_i}=(R+r)\dot{\vec\gamma}_i$ is its velocity, $i=1,\dots,n$.
The kinetic energy plays the role of the Lagrangian.

Let us denote the contact points of the balls $\mathbf B_1,\dots,\mathbf B_n$ with the spheres $\mathbf S_0$ and $\mathbf S$ by $A_1,..., A_n$ and
$B_1, B_2,...,B_n$, respectively.
The condition that the rolling of the balls $\mathbf B_1,\dots,\mathbf B_n$ and the sphere $\mathbf S$ are without slipping leads to the nonholonomic constraints:
\[
\vec{v}_{O_i}+\vec{\omega}_i\times\overrightarrow{O_iA_i}=0, \qquad \vec{v}_{O_i}+\vec{\omega}_i\times\overrightarrow{O_iB_i}=\vec{\omega}\times\overrightarrow{OB_i},
\qquad i=1,...,n,
\]
that is,
\begin{equation}\label{VEZE}
\vec{v}_{O_i}=r\vec{\omega}_i\times\vec{\gamma}_i, \qquad \vec{v}_{O_i}=(R+2r)\vec{\omega}\times\vec{\gamma}_i-r\vec{\omega}_i\times\vec{\gamma}_i,\qquad  i=1,...,n.
\end{equation}

The dimension of the configuration space $Q$ is $5n+3$. There are $4n$ independent constraints in \eqref{VEZE}, defining
a nonintegrable distribution  $\mathcal D\subset TQ$. Therefore, the dimension of the vector subspaces of admissible
velocities $\mathcal D_q\subset T_q Q$ is $n+3$, $q\in Q$.
The phase space of the system has the dimension $6n+6$, which is the dimension of the bundle $\mathcal D$ as a submanifold of $TQ$.

The equations of motion of the spherical ball bearing problem are given by the Lagrange-d'Alembert equations {\cite{AKN, Bloch}}
\begin{equation}\label{L1}
\delta T=\big(\frac{\partial T}{\partial q}-\frac{d}{dt}\frac{\partial T}{\partial \dot q},\delta q\big)=0, \quad \text{for all virtual displacement} \quad \delta q\in\mathcal D_q.
\end{equation}

Instead of using the Lagrange-d'Alembert equations,
below we will derive the equations of motion directly from the fundamental laws of classical mechanics and using the vector notation.

\begin{remark}
We will show in Proposition \ref{prva} that if the initial conditions are chosen such that the distances between $O_i$ and $O_j$ are all greater than $2r$, $1 \le i< j \le n$, then
the balls will not have collisions along the course of motion. This is the reason why we do not assume additional one-side constraints
\begin{equation}\label{oblast}
\vert \vec\gamma_i-\vec\gamma_j \vert \ge \frac{2r}{r+R}, \qquad 1\le i<j\le n.
\end{equation}
\end{remark}

The kinetic energy and the constraints are invariant with respect to the $SO(3)^{n+1}$--action defined by
\begin{equation}\label{transf}
(\mathbf g,\mathbf g_1,\dots,\mathbf g_n,\vec\gamma_1,\dots,\vec\gamma_n)\longmapsto (\mathbf a\mathbf g,\mathbf a\mathbf g_1{\mathbf a}_1^{-1},\dots,
\mathbf a\mathbf g_n{\mathbf a}_n^{-1},\mathbf a\vec\gamma_1,\dots,\mathbf a\vec\gamma_n),
\end{equation}
$\mathbf a,\mathbf a_1,\dots,\mathbf a_n\in SO(3)$, representing a freedom in the choice of the reference frames
\[
O\vec{\mathbf e}^0_1,\vec{\mathbf e}^0_2,\vec{\mathbf e}^0_3,
\qquad O_i\vec{\mathbf e}^i_1,\vec{\mathbf e}^i_2,\vec{\mathbf e}^i_3, \qquad i=1,\dots,n.
\]

Indeed, for the extension of the transformation \eqref{transf} to the tangent bundle $TQ$ we have
\begin{align*}
\Omega=\mathbf g^{-1}\dot{\mathbf g}      &\longmapsto \big(\mathbf a\mathbf g\big)^{-1}\mathbf a\dot{\mathbf g}=\Omega, \\
\omega=\dot{\mathbf g}{\mathbf g}^{-1}    &\longmapsto \mathbf a\dot{\mathbf g}\big(\mathbf a{\mathbf g}\big)^{-1}=\mathbf a\omega\mathbf a^{-1},\\
\omega_i=\dot{\mathbf g}_i \mathbf g_i^{-1}&\longmapsto \mathbf a \dot{\mathbf g}_i \mathbf a_i^{-1} \big(\mathbf a\mathbf g_i\mathbf a_i^{-1}\big)^{-1}=\mathbf a\omega_i\mathbf a^{-1},\\
\dot{\vec\gamma}&\longmapsto \mathbf a \dot{\vec\gamma},
\end{align*}
and, therefore,
\begin{equation}\label{transf1}
\vec{\Omega}\longmapsto \vec{\Omega}, \qquad \vec{\omega}\longmapsto \mathbf a\vec{\omega}, \qquad \vec{\omega}_i\longmapsto \mathbf a\vec{\omega}_i, \qquad
\vec{v}_{O_i}\longmapsto \mathbf a\vec{v}_{O_i}.
\end{equation}

It is clear that the kinetic energy and the constraints are invariant with respect to the transformation \eqref{transf1}.
Also, note that \eqref{transf} does not change the vectors $\vec{\omega}_i$, $\vec{\gamma}_i$, $\vec{v}_{O_i}$ written in the moving frame $O\vec{\mathbf e}_1,\vec{\mathbf e}_2,\vec{\mathbf e}_3$:
\begin{align*}
\vec{\Omega}_i=\mathbf g^{-1}\vec{\omega_i}  &\longmapsto \big(\mathbf a\mathbf g\big)^{-1}\mathbf a\vec{\omega_i}=\vec{\Omega}_i,  \\
\vec{\Gamma}_i=\mathbf g^{-1} \vec{\gamma}_i &\longmapsto \big(\mathbf a\mathbf g\big)^{-1}\mathbf a \vec{\gamma}_i=\vec{\Gamma}_i
\quad (\text{implying   }\dot{\vec\Gamma}_i\longmapsto  \dot{\vec\Gamma}_i), \\
\vec{V}_{O_i}=\mathbf g^{-1}\vec{v}_{O_i}=(R+r)\mathbf g^{-1}\dot{\vec\gamma}_i
&\longmapsto  (R+r)\big(\mathbf a\mathbf g\big)^{-1}\mathbf a\dot{\vec\gamma}_i=\mathbf g^{-1}\vec{v}_{O_i}=\vec{V}_{O_i}.
\end{align*}

Thus, for the coordinates in the space $(TQ)/SO(3)^{n+1}$ we can take the angular velocities and the unit position vectors in the reference frame
attached to the sphere $\mathbf S$:
\[
(TQ)/SO(3)^{n+1}\cong{\mathbb R}^{3(n+1)}\times (TS^2)^n \{\vec\Omega,\vec{\Omega}_1,\dots,\vec{\Omega}_n,\dot{\vec{\Gamma}}_1,\dots,\dot{\vec{\Gamma}}_n,\vec{\Gamma}_1,\dots,\vec{\Gamma}_n\}.
\]

In the moving reference frame $O\vec{\mathbf e}_1,\vec{\mathbf e}_2,\vec{\mathbf e}_3$, the constraints become:
\begin{align}
\label{veze1} \vec{V}_{O_i}&=(R+2r)\vec{\Omega}\times\vec{\Gamma}_i-r\vec{\Omega}_i\times\vec{\Gamma}_i,\\
\label{veze1*}\vec{V}_{O_i}&=r\vec{\Omega}_i\times\vec{\Gamma}_i, \qquad\qquad\qquad i=1,\dots,n,
\end{align}
defining the \emph{reduced phase space} $\mathcal M=\mathcal D/SO(3)^{n+1}\subset (TQ)/SO(3)^{n+1}$ of dimension $3n+3$.

Since both the kinetic energy and the constraints are invariant with respect to the $SO(3)^{n+1}$--action \eqref{transf}, the equations of motion \eqref{L1} are also
$SO(3)^{n+1}$--invariant. Thus, they induce a well defined system on the reduced phase space $\mathcal M$. 

\section{ The kinematic and momentum equations in the moving frame}\label{sec3}

The time derivative of $\vec{\Gamma}_i$ can be directly extracted from the constraints as follows.

\begin{lem} The kinematic part of the equations of motion of the spherical ball bearing system is:
\begin{equation}\label{game1}
\dot{\vec{\Gamma}}_i=\frac{R}{2R+2r}\vec{\Gamma}_i\times\vec{\Omega}, \qquad i=1,\dots,n.
\end{equation}
\end{lem}

\begin{proof}
Let us at the moment consider the fixed reference frame $O\vec{\mathbf e}^0_1,\vec{\mathbf e}^0_2,\vec{\mathbf e}^0_3$. One has
\[
\dot{\overrightarrow{OO}}_i+\vec{\omega}_i\times\overrightarrow{O_iA_i}=0.
\]

Therefore,
$(R+r)\dot{\vec{\gamma}}_i-r\vec{\omega}_i\times\vec{\gamma}_i=0$, or equivalently
\begin{equation}\label{eq:gamma_i}
\dot{\vec{\gamma}}_i=\frac{r}{R+r}\vec{\omega}_i\times\vec{\gamma}_i.
\end{equation}

The equation \eqref{eq:gamma_i} in the moving reference frame $O\vec{\mathbf e}_1,\vec{\mathbf e}_2,\vec{\mathbf e}_3$ has
the form
\[
\dot{\vec{\Gamma}}_i+\vec{\Omega}\times\vec{\Gamma}_i=\frac{r}{R+r}\vec{\Omega}_i\times\vec{\Gamma}_i.
\]

Thus, we get
\begin{equation}\label{game}
\dot{\vec{\Gamma}}_i=\big(\frac{r}{R+r}\vec{\Omega}_i-\vec{\Omega}\big)\times\vec{\Gamma}_i.
\end{equation}

From the constraints \eqref{veze1}, \eqref{veze1*}, we obtain
\begin{equation}\label{veze2}
\vec{\Omega}_i\times\vec{\Gamma}_i=\frac{R+2r}{2r}\vec{\Omega}\times\vec{\Gamma}_i, \qquad i=1,\dots,n.
\end{equation}
Finally, using \eqref{veze2}, the equations \eqref{game} can be written in a more convenient form \eqref{game1}.
\end{proof}

As a consequence, we have:

\begin{prop}\label{prva} The following functions are the first integrals of  motion:
\[
\langle\vec{\Gamma}_i,\vec{\Gamma}_j\rangle=\gamma_{ij}=const, \qquad i,j=1,\dots,n.
\]
\end{prop}

\begin{proof} By a direct differentiation, we get
\begin{align*}
\frac{d}{dt}\langle\vec{\Gamma}_i,\vec{\Gamma}_j\rangle &=
\langle\dot{\vec{\Gamma}}_i,\vec{\Gamma}_j\rangle+
\langle\vec{\Gamma}_i,\dot{\vec{\Gamma}}_j\rangle\\
&=\frac{R}{2R+2r}\big(\langle \vec{\Gamma}_i\times\vec{\Omega},\vec{\Gamma}_j\rangle+\langle\vec{\Gamma}_j\times\vec{\Omega},\vec{\Gamma}_i\rangle\big)=0.
\end{align*}
\end{proof}

In other words, the centers $O_i$ of the homogeneous balls $\mathbf B_i$ are in rest in relation to each other.
In particular, since $\langle \vec\gamma_i,\vec\gamma_j\rangle=\langle \vec\Gamma_i,\vec\Gamma_j\rangle$,
the interior of the region \eqref{oblast} is invariant under the flow of the system.

Next, let $\vec{\mathbf F}_{B_i}$ and $\vec{\mathbf F}_{A_i}$ be the reaction forces that act on the ball $\mathbf B_i$ at the points $B_i$ and $A_i$, respectively.
The reaction force at the point $B_i$ on the sphere $\mathbf S$ is then $-\vec{\mathbf F}_{B_i}$.

By using the laws of change of angular momentum and  momentum of a rigid body in the moving reference frame for the balls $\mathbf B_i$ and the sphere $\mathbf S$ {(e.g., see \cite{AKN})}, we get:

\begin{lem} The dynamical part of the equations of motion of the spherical ball bearing system is:
\begin{align}
\label{jednacine1} I_i\dot{\vec{\Omega}}_i &=I_i\vec{\Omega}_i\times \vec{\Omega} + r\vec{\Gamma}_i\times(\vec{\mathbf F}_{B_i}-\vec{\mathbf F}_{A_i}),\\
\label{jednacine2} m_i\dot{\vec{V}}_{O_i} &=m_i\vec{V}_{O_i} \times \vec{\Omega}+\vec{\mathbf F}_{B_i}+\vec{\mathbf F}_{A_i}, \qquad\qquad\qquad i=1,...,n\\
\label{jednacine3} I\dot{\vec{\Omega}}&=I\vec{\Omega}\times \vec{\Omega} -\sum_{i=1}^{n}(R+2r)\vec{\Gamma}_i\times\vec{\mathbf F}_{B_i}.
\end{align}
\end{lem}

For $n=1$ and the absence of the interior fixed sphere $\mathbf S_0$, i.e. $\vec{\mathbf F}_{A_i}=0$, see \cite{BKM}.

We still need to calculate the torques of reaction forces. Prior to that,
we formulate and prove the following important statement.

\begin{prop}\label{go} The projections of the angular velocities $\vec{\Omega}_i$ to
to the directions $\vec{\Gamma}_i$ are the first integrals of motion:
\[
\langle\vec{\Omega}_i,\vec{\Gamma}_i\rangle=c_i=const, \qquad i=1,...,n.
\]
\end{prop}
\begin{proof}
From \eqref{jednacine1} and \eqref{game} we get
\begin{align*}
\frac{d}{dt}\langle\vec{\Omega}_i,\vec{\Gamma}_i{\rangle} =&\langle\dot{\vec{\Omega}}_i,\vec{\Gamma}_i\rangle+\langle\vec{\Omega}_i,\dot{\vec{\Gamma}}_i\rangle\\
=  &\langle\vec{\Omega}_i\times\vec{\Omega}, \vec{\Gamma}_i\rangle+\langle \frac{r}{I_i}\vec{\Gamma}_i\times(\vec{\mathbf F}_{B_i}-\vec{\mathbf F}_{A_i}),\vec\Gamma_i \rangle\\ &+\langle\vec{\Omega}_i,\frac{r}{R+r}\vec{\Omega}_i\times \vec{\Gamma}_i\rangle -\langle\vec{\Omega}_i,\vec{\Omega}\times \vec{\Gamma}_i\rangle=0.
\end{align*}
\end{proof}

\section{ The reduced system}\label{sec4}

From the constraints written in the moving frame \eqref{veze1}, \eqref{veze1*}, \eqref{veze2}, we get
\[
\langle\vec{\Omega}\times\vec{\Gamma}_i, \vec{\Omega}_i\rangle=0.
\]
That means that vectors $\vec{\Gamma}_i,\ \vec{\Omega},\ \vec{\Omega}_i$ are coplanar. Moreover, we obtain:
\begin{equation}\label{omegai}
\vec{\Omega}_i=\langle\vec{\Gamma}_i,\vec{\Omega}_i\rangle\vec{\Gamma}_i+\frac{R+2r}{2r}\vec{\Omega}-\frac{R+2r}{2r}\langle\vec{\Gamma}_i,\vec{\Omega}\rangle\vec{\Gamma}_i.
\end{equation}

Further, from Proposition \ref{go}, we get that the reduced phase space $\mathcal M=\mathcal D/SO(3)^{n+1}$ is foliated on $2n+3$--dimensional invariant varieties
\[
\mathcal M_c: \qquad \langle\vec{\Omega}_i,\vec{\Gamma}_i\rangle=c_i=const, \qquad i=1,...,n.
\]

On the invariant variety $\mathcal M_c$, the vector-functions
$\vec{\Omega}_i$ can be uniquely expressed as functions of $\vec{\Omega}$, $\vec{\Gamma}_i$ using the equation \eqref{omegai},:
\begin{equation}\label{omegai1}
\vec{\Omega}_i=c_i\vec{\Gamma}_i+\frac{R+2r}{2r}\vec{\Omega}-\frac{R+2r}{2r}\langle\vec{\Gamma}_i,\vec{\Omega}\rangle\vec{\Gamma}_i.
\end{equation}

 Whence, $\vec\Omega$ determines all velocities of the system on $\mathcal M_c$ and $\mathcal M_c$ is diffeomorphic
to the \emph{second reduced phase space}
\[
\mathcal N=\R^3\times \big(S^2\big)^{n}\{\Omega,\vec\Gamma_1,\dots,\vec\Gamma_n\}.
\]

This can be seen as follows. Consider the natural projection
\begin{align*}
&\pi\colon (TQ)/SO(3)^{n+1}\cong{\mathbb R}^{3(n+1)}\times (TS^2)^n \to \mathcal N,\\
&\pi(\vec\Omega,\vec{\Omega}_1,\dots,\vec{\Omega}_n,\dot{\vec{\Gamma}}_1,\dots,\dot{\vec{\Gamma}}_n,\vec{\Gamma}_1,\dots,\vec{\Gamma}_n)=(\vec\Omega,\vec{\Gamma}_1,\dots,\vec{\Gamma}_n),
\end{align*}
and let $\pi_c$ be the restriction to
$\mathcal M_c \subset \mathcal M \subset (TQ)/SO(3)^{n+1}$ of $\pi$.
Then the projection
\[
\pi_c\colon \mathcal M_c\longmapsto \mathcal N
\]
is a bijection.

Thus, instead of the derivation of the torques of all reaction forces in \eqref{jednacine1} and \eqref{jednacine3}, it is sufficient to
find the torque in the equation \eqref{jednacine3} on a given invariant variety $\mathcal M_c$.

To simplify the equations \eqref{game1} and \eqref{omegai1}, we introduce the parameters
\begin{equation}\label{parametri}
\varepsilon=\frac{R}{2R+2r} \qquad \text{and} \qquad \delta=\frac{R+2r}{2r}.
\end{equation}

We define the \emph{modified operator of inertia} $\mathbf I$ as
\begin{equation}\label{modI}
\begin{aligned}
\mathbf I=I+\delta^2\sum_{i=1}^n(I_i+m_ir^2)\pr_i,
\end{aligned}
\end{equation}
where $\pr_i\colon \R^3 \to \vec{\Gamma}_i^\perp$ is the orthogonal projection to the plane orthogonal to $\vec{\Gamma}_i$.
We set
\begin{align}
\label{m1}\vec{M}=&\, \mathbf I\vec\Omega=
I\vec\Omega+\delta^2\sum_{i=1}^n(I_i+m_ir^2)\vec{\Omega}-\delta^2\sum_{i=1}^n(I_i+m_ir^2)\langle\vec{\Gamma}_i,\vec{\Omega}\rangle\vec{\Gamma}_i,\\
\label{n1}\vec{N}=&\, \delta\sum_{i=1}^n I_ic_i\vec{\Gamma}_i.
\end{align}

\begin{thm}\label{Glavna}
The reduction of the spherical ball bearing problem to $\mathcal M_c\cong \mathcal N$ is described
by the equations
\begin{align}
\label{red1} &\frac{d}{dt}{\vec{M}}=\vec{M}\times\vec{\Omega}+(1-\varepsilon)\vec N\times\vec{\Omega},\\
\label{red2} &\frac{d}{dt}\vec\Gamma_i=\varepsilon\vec\Gamma_i\times\vec\Omega, \qquad\qquad i=1,\dots,n.
\end{align}
\end{thm}

Note that the kinetic energy of the system takes the form
\[
T=\frac12 \langle \vec M,\vec\Omega\rangle+\frac12\sum_{i=1}^n I_i c_i^2.
\]

Also, since
\[
\frac{d}{dt}\vec N=\varepsilon \vec N\times \vec\Omega,
\]
the equation \eqref{red1} is equivalent to
\begin{equation}
\frac{d}{dt}(\vec M+\vec N)=(\vec M+\vec N)\times\vec\Omega.
\label{red1*}
\end{equation}

\begin{proof}[Proof of Theorem \ref{Glavna}]
From the equations \eqref{jednacine1} and \eqref{jednacine2} one have
\[
\vec{\Gamma}_i\times\vec{\mathbf F}_{B_i}=\frac1{2r}(I_i\dot{\vec{\Omega}}_i+\vec{\Omega}\times(I_i\vec{\Omega}_i))+\frac{m_i}{2}\vec{\Gamma}_i\times\dot{\vec{V}}_{O_i}+
\frac{m_i}{2}\vec{\Gamma}_i\times(\vec{\Omega}\times\vec{V}_{O_i})
\]
By plugging the last expression in the third equation of motion \eqref{jednacine3}, it becomes
\begin{equation}
\label{jednacina3}
\begin{aligned}
I\dot{\vec{\Omega}}+\vec{\Omega}\times I\vec{\Omega}=-&\sum_{i=1}^{n}\Big[\frac{R+2r}{2r}(I_i\dot{\vec{\Omega}}_i +\vec{\Omega}\times(I_i\vec{\Omega}_i))+\\
&\frac{m_i(R+2r)}{2}\vec{\Gamma}_i\times\dot{\vec{V}}_{O_i}+\frac{m_i(R+2r)}{2}\vec{\Gamma}_i\times(\vec{\Omega}\times\vec{V}_{O_i})\Big].
\end{aligned}
\end{equation}

From \eqref{veze1*}, \eqref{game1}, and \eqref{game}, we get
$\dot{\vec{\Gamma}}_i\times\vec{V}_{O_i}=0$, and, therefore
\[
\frac{d}{dt}\big(\vec{\Gamma}_i\times\vec{V}_{O_i}\big)=\vec{\Gamma}_i\times\dot{\vec{V}}_{O_i}.
\]
Also, we have
\[
\vec{\Gamma}_i\times(\vec{\Omega}\times\vec{V}_{O_i})=\vec{\Omega}\times(\vec{\Gamma}_i\times\vec{V}_{O_i}).
\]

Having in mind the last two expressions, the equation \eqref{jednacina3} becomes
\begin{equation}
\label{jednacina3a}
\begin{aligned}
\frac{d}{dt}&\Big(I\vec{\Omega}+\sum_{i=1}^{n}\big(\frac{R+2r}{2r}I_i\vec{\Omega}_i+\frac{m_i(R+2r)}{2}\vec{\Gamma}_i\times{\vec{V}}_{O_i}\big)\Big)=\\
&-\vec{\Omega}\times \Big(I\vec{\Omega}+\sum_{i=1}^{n}\big(\frac{R+2r}{2r}I_i\vec{\Omega}_i+\frac{m_i(R+2r)}{2}\vec{\Gamma}_i\times{\vec{V}}_{O_i}\big)\Big)
\end{aligned}
\end{equation}

Finally, using \eqref{omegai1}, constraints \eqref{veze2}, the definitions \eqref{parametri}, \eqref{m1}, and \eqref{n1} of parameters $\varepsilon$ and $\delta$ and
the vectors  $\vec M$ and $\vec N$, the equation \eqref{jednacina3a} takes the form \eqref{red1*}.
\end{proof}

\begin{remark}
If we formally set $\varepsilon=1$ in the system \eqref{red1}, \eqref{red2} we obtain the equation of the spherical support system introduced by
Fedorov in \cite{F1}. The system describes the rolling without slipping of a dynamically nonsymmetric sphere $\mathbf S$ over $n$ homogeneous balls $\mathbf B_1, \dots,\mathbf B_n$ of possibly different radii, but with fixed centers. It is an example of a class of
nonhamiltonian L+R systems on Lie groups with an invariant measure (see \cite{FeRCD, FJ1, Jo1}).
On the other hand, if we set $\vec N=0$, we obtain an example $\varepsilon$--modified L+R system studied in \cite{Jo2015}.
\end{remark}

Since $\langle\vec{\omega}_i,\vec{\gamma}_i\rangle=\langle\vec{\Omega}_i,\vec{\Gamma}_i\rangle$, we have that $\mathcal D$ is foliated
on invariant varieties
\[
\mathcal D_c\colon \qquad \langle\vec{\omega}_i,\vec{\gamma}_i\rangle=c_i, \qquad i=1,\dots,n, \qquad \dim\mathcal D_c=5n+6
\]
and $\mathcal M_c=\mathcal D_c/SO(3)^{n+1}$.
As a result we obtain  the following diagram
\begin{equation*}
\label{principal}
\xymatrix@R28pt@C28pt{
\mathcal D_c \,\,\ar@{->}[d]^{/SO(3)^{n+1}} \ar@{^{(}->}[r]& \,\, \mathcal D\ar@{^{}->}[d]^{/SO(3)^{n+1}} \,\, \ar@{^{(}->}[r] & TQ=(TSO(3))^{n+1}\times (TS^2)^n \ar@{^{}->}[d]^{/SO(3)^{n+1}}\\
\mathcal M_c \,\,\ar@{->}[rrd]^{\pi_c}_{\cong} \ar@{^{(}->}[r]& \,\,\mathcal M\,\,  \ar@{^{(}->}[r] & (TQ)/SO(3)^{n+1}\cong{\mathbb R}^{3(n+1)}\times (TS^2)^n \ar@{^{}->}[d]^{\pi}  \\
 &  & \mathcal N=\R^3\times \big(S^2\big)^{n}  }
\end{equation*}
which implies that $\mathcal D_c$ and $\R^3\times SO(3)^{n+1}\times\big(S^2\big)^{n}\{\Omega,\mathbf g,\mathbf g_1,\dots,\mathbf g_n,\vec\Gamma_1,\dots,\vec\Gamma_n\}$ are diffeomorphic:
\[
\mathcal D_c\cong \R^3\times SO(3)^{n+1}\times\big(S^2\big)^{n}\{\Omega,\mathbf g,\mathbf g_1,\dots,\mathbf g_n,\vec\Gamma_1,\dots,\vec\Gamma_n\}.
\]

\begin{cor} \label{Originalne}
The complete equations of  motion of the sphere $\mathbf S$ and the balls $\mathbf B_1,\dots,\mathbf B_n$ of the spherical ball bearing problem
on the invariant manifold $\mathcal D_c$ are given by
\begin{align*}
\dot{\vec{M}}&=\vec{M}\times\vec{\Omega}+(1-\varepsilon)\vec N\times\vec{\Omega},\\
\dot{\mathbf g}&=\mathbf g \Omega,\\
\dot{\mathbf g}_i&=\mathbf g\Omega_i(\vec\Omega,\vec\Gamma_i,c_i)\mathbf g_i,\\
\dot{\vec\Gamma}_i&=\varepsilon\vec\Gamma_i\times\vec\Omega, \qquad\qquad i=1,\dots,n,
\end{align*}
where $\vec M$, $\vec N$, are given by \eqref{m1} and \eqref{n1}. Here $\Omega$ and $\Omega_i(\vec\Omega,\vec\Gamma_i,c_i)$ are skew-symmetric matrices related to $\vec\Omega$ and $\vec{\Omega}_i$
 after the identification \eqref{izomorfizam}; $\vec{\Omega}_i=\vec\Omega_i(\vec\Omega,\vec\Gamma_i,c_i)$ as in the equation \eqref{omegai1}.
\end{cor}

\section{The associated system on $\R^3\times Sym(3)$ and an invariant measure}\label{sec5}

Let
\[
\Gamma=-\delta^2\sum_{i=1}^n(I_i+m_ir^2)\pr_i
\]
be the symmetric operator using which the definition of the modified inertia operator $\mathbf I$   \eqref{modI} can be rewritten as:
\[
\mathbf I=I-\Gamma, \qquad \Gamma=\delta^2\sum_{i=1}^n(I_i+m_ir^2)\big(\vec{\Gamma}_i\otimes\vec{\Gamma}_i-\mathbf E\big), \qquad \mathbf E=\diag(1,1,1).
\]

Along the flow of the system, $\Gamma$ satisfies the equation
\begin{equation}\label{GAMA}
\frac{d}{dt}\Gamma=\varepsilon[\Gamma,\Omega],
\end{equation}
where $\Omega$ is the skew-symmetric matrix that corresponds to the angular velocity $\vec\Omega$ via isomorphism \eqref{izomorfizam}.

Let us consider a special case when $c_1=0,\dots,c_n=0$, i.e., the invariant manifold $\mathcal M_0$. This means that there are no twisting of the balls, i.e. the vectors $\vec{\Omega}_i$ and $\vec{\Gamma}_i$ are orthogonal to each other.
However, note that this conditions are not nonholonomic constraints, but the first integrals of motion.

As a result we obtain the associated system
\begin{equation}\label{eqc0}
\begin{aligned}
\dot{\vec{M}}&=\vec{M}\times\vec{\Omega},   \qquad \vec M=\mathbf I\vec\Omega=I\vec\Omega-\Gamma\vec\Omega,\\
\dot{\Gamma}&=\varepsilon[\Gamma,\Omega]
\end{aligned}
\end{equation}
on the space $\R^3\times Sym(3)$, where $Sym(3)$ are $3\times 3$ symmetric matrices.
The system belongs to the class of $\varepsilon$--modified L+R systems studied in \cite{Jo2015}.

Let $d\Omega$ and $d\Gamma$ be the standard measures on $\R^3\{\vec\Omega\}$ and $Sym(3)\{\Gamma\}$.
The system \eqref{eqc0} possesses the invariant measure
$\mu(\Gamma) d\Omega \wedge d\Gamma$
with the density $\mu(\Gamma)=\sqrt{\det(\mathbf I)}$ (see Theorem 4, \cite{Jo2015})
Therefore, $\mu=\sqrt{\det(\mathbf I)}$ is a natural candidate for the density of an invariant measure of the system
\eqref{red1}, \eqref{red2} when the constants $c_i$ are different from zero.
Indeed, we have

\begin{thm}\label{mera}
For arbitrary values of parameters $c_i$, the reduced system \eqref{red1}, \eqref{red2} has the invariant measure
\begin{equation}\label{MERA}
\mu(\vec\Gamma_1,\dots,\vec\Gamma_n)d\Omega\wedge \sigma_1\wedge \dots \wedge \sigma_n, \qquad \mu=\sqrt{\det(\mathbf I)}=\sqrt{\det(I-\Gamma)},
\end{equation}
where $d\Omega$ and $\sigma_i$ are the standard measures on $\R^3\{\vec\Omega\}$ and $S^2\{\vec\Gamma_i\}$, $i=1,\dots,n$.
\end{thm}

The proof of the Theorem we are going to present is a variant of a corresponding proof for $\varepsilon$--modified L+R systems.
It is given below for the completeness of the exposition.
In what follows we use

\begin{lem}\label{lemica}
Let $A$ be a symmetric matrix and let $\vec\Omega\in\R^3$ and $\Omega\in so(3)$ be related by \eqref{izomorfizam}. Then:
\begin{itemize}
\item[(i)] the symmetric part of the matrix $\partial\big(A\vec{\Omega}\times\vec{\Omega}\big)/{\partial\vec{\Omega}}$ is equal to $\frac{1}{2}[A,\Omega]$;
\item[(ii)] $A\vec{\Omega}\times\vec{\Omega}=[A,\Omega]\vec{\Omega}$.
\end{itemize}
\end{lem}

\begin{proof}[Proof of Theorem \ref{mera}]
We can consider the system
\begin{align}
\label{pros1} &\frac{d}{dt}\big({\mathbf I\vec\Omega}\big)=\mathbf I\vec{\Omega}\times\vec{\Omega}+(1-\varepsilon)\vec N\times\vec{\Omega},\qquad \mathbf I=I-\Gamma,\\
\label{pros2} &\frac{d}{dt}\vec\Gamma_i=\varepsilon\vec\Gamma_i\times\vec\Omega, \qquad\qquad\qquad\qquad\qquad i=1,\dots,n,
\end{align}
as extended in the Euclidean space $\R^{3n+3}\{\vec\Omega,\vec\Gamma_1,\dots,\vec\Gamma_n\}$ as well.
The extended system also has first integrals $\langle\vec\Gamma_i,\vec\Gamma_j\rangle=\gamma_{ij}$. In particular, by taking $\gamma_{ii}=1$, $i=1,\dots,n$,
we get that the reduced system is the restriction of the extended system \eqref{pros1}, \eqref{pros2} from $\R^{3n+3}$ to the invariant variety $\mathcal N$.

Therefore, it is sufficient to prove that the extended system preserves the measure
\[
\nu=\mu(\vec\Gamma_1,\dots,\vec\Gamma_n)d\Omega\wedge d\Gamma_1\wedge \dots \wedge d\Gamma_n, \qquad \mu=\sqrt{\det(\mathbf I)},
\]
where $d\Gamma_i$ is the standard measure in $\R^3\{\vec\Gamma_i\}$.

This is a standard construction: if a system has an invariant measure, then the restriction of the system to an invariant manifold also has an invariant measure induced by the
measure from the ambient space. Let
\[
X=(\dot{\vec{\Omega}},\dot{\vec\Gamma}_1,\dots,\dot{\vec\Gamma}_n)
\]
be the vector field on $\R^{3n+3}$ defined by the equations \eqref{pros1}, \eqref{pros2} and assume that the Lie derivative $\mathcal L_X$ of $\nu$ vanishes.
It is well known that for $\vec\Gamma_i\ne 0$, the volume form in $\R^3\{\vec\Gamma_i\}$ can be written as
$d\Gamma_i=\alpha_i\wedge \sigma_i$, where $\sigma_i$
is the standard measure on the unit sphere and
\[
\alpha_i=\vert\vec\Gamma_i\vert^2 d(\vert \vec\Gamma_i\vert)=\frac13 d\langle\vec\Gamma_i,\vec\Gamma_i\rangle^{\frac32}, \qquad i=1,\dots,n.
\]

Since $\mathcal L_X(\alpha_i)=0$, we have
\[
\mathcal L_X(\nu)=const \cdot \alpha_1 \wedge \dots \wedge \alpha_n \wedge
\mathcal L_X\big(\mu\, d\Omega \wedge \sigma_1\wedge \dots \wedge \sigma_n\big)=0,
\]
$\vec\Gamma_i\ne 0$, $i=1,\dots,n$. Therefore, the reduced system preserves the measure \eqref{MERA}.

The equation $\mathcal L_X(\nu)=0$ in $\R^{3n+3}$ can be written in the equivalent form
\begin{equation}\label{uslov}
\dot{\mu} +\mu\dv(X)=\dot{\mu}
+\mu\tr\frac{\partial{\dot{\vec\Omega}}}{\partial\vec\Omega} + \mu\sum_{i=1}^n\tr\frac{\partial{\dot{\vec\Gamma}_i}}{\partial\vec\Gamma_i}=0.
\end{equation}

We have the following equalities
\begin{align*}
\dot{\mathbf I}\vec{\Omega}&+{\mathbf I}\dot{\vec{\Omega}}=I\vec{\Omega}\times\vec{\Omega}-\Gamma\vec{\Omega}\times\vec{\Omega}+(1-\varepsilon)\vec N\times\vec{\Omega},\\
{\mathbf I}\dot{\vec{\Omega}}&=I\vec{\Omega}\times\vec{\Omega}-\Gamma\vec{\Omega}\times\vec{\Omega}+\varepsilon[\Gamma,\Omega]\vec{\Omega}+(1-\varepsilon)\vec N\times\vec{\Omega}\\
&=I\vec{\Omega}\times\vec{\Omega}-(1-\varepsilon)\Gamma\vec{\Omega}\times\vec{\Omega}+(1-\varepsilon)\vec N\times\vec{\Omega}, \qquad \text{(item (ii) of Lemma \ref{lemica})}.
\end{align*}

Thus,
\begin{align*}
&\dot{\vec{\Omega}}={\mathbf I}^{-1}\big(I\vec{\Omega}\times\vec{\Omega}-(1-\varepsilon)\Gamma\vec{\Omega}\times\vec{\Omega}+(1-\varepsilon)\vec N\times\vec{\Omega}\big),\\
&\tr\frac{\partial{\dot{\vec\Omega}}}{\partial\vec\Omega}=\tr\Big({\mathbf I}^{-1}\frac{\partial}{\partial\vec{\Omega}}\big(I\vec{\Omega}\times\vec{\Omega}-(1-\varepsilon)\Gamma\vec{\Omega}\times\vec{\Omega}+(1-\varepsilon)\vec N\times\vec{\Omega}\big)\Big).
\end{align*}

The matrix $\partial(\vec N\times\vec\Omega)/\partial\vec\Omega$ is skew-symmetric.
Since ${\mathbf I}^{-1}$ is symmetric, only the symmetric part of the expression in parenthesis in the last equation matters.
Whence, using item (i) of Lemma \ref{lemica}, we get
\begin{equation}\label{div1}
\begin{aligned}
\tr\frac{\partial{\dot{\vec\Omega}}}{\partial\vec\Omega}&=\tr\Big({\mathbf I}^{-1}\big(\frac{1}{2}[I,\Omega]-\frac{1-\varepsilon}{2}[\Gamma,\Omega]\big)\Big)\\
&=\tr\Big({\mathbf I}^{-1}\frac{1}{2}[{\mathbf I},\Omega]+\frac{\varepsilon}{2}{\mathbf I}^{-1}[\Gamma,\Omega]\Big)=\frac{\varepsilon}{2}\tr({\mathbf I}^{-1}[\Gamma,\Omega]).
\end{aligned}
\end{equation}

On the other hand
\begin{equation}\label{div2}
\begin{aligned}
\dot{\mu}=&\frac{1}{2\sqrt{\det ({\mathbf I})}}\frac{d}{dt}\det({\mathbf I})=
\frac{1}{2\sqrt{\det ({\mathbf I})}}\det({\mathbf I})\tr\big({\mathbf I}^{-1}\frac{d}{dt}\big(I-\Gamma\big)\big)\\
=&         -\frac{1}{2}\mu\tr({\mathbf I}^{-1}\varepsilon[\Gamma,\Omega]),
\end{aligned}
\end{equation}
Here we used a well-known formula $\frac{d}{dt}\det({\mathbf I})=\det({\mathbf I})\tr({\mathbf I}^{-1}\dot{\mathbf I})$.

Since the matrices ${\partial{\dot{\vec\Gamma}_i}}/{\partial\vec\Gamma_i}$ are skew symmetric and have zero traces,
the equations \eqref{div1} with \eqref{div2} imply the required condition \eqref{uslov}.
\end{proof}

Note that the existence of an invariant measure for nonholomic problems is well studied in many classical problems \cite{BM2002, BMB2013}. After Kozlov's theorem on obstruction to the existence of an
invariant measure for the variant of the classical Suslov problem (e.g., see \cite{BT2020, FJ1}) on Lie algebras \cite{Kozlov},  general existence statements
for nonholonomic systems with  symmetries are obtained in  \cite{ZenkovBloch} and \cite{FGM2015}.

A closely related problem is the integrability of the nonholonomic systems \cite{AKN}. Here we have the following statement.

\begin{prop}
The system \eqref{red1}, \eqref{red2} always has the following first integrals
\[
F_1=\frac12 \langle \vec M,\vec\Omega\rangle, \quad F_2=\langle \vec M+\vec N, \vec M+\vec N\rangle, \quad
F_{ij}=\langle \vec\Gamma_i, \vec\Gamma_j\rangle, \quad 1\le i < j\le n.
\]
\end{prop}

Thus, in the special case $n=1$, we have the 5-dimensional phase space $\mathcal N=\R^3\times S^2\{\vec\Omega,\vec\Gamma_1\}$, and the system has
two first integrals and an invariant measure. For the integrability,  one needs to find a third independent first integral.
We will study integrability in the spherical ball bearing problems in a separate paper.
Also, it would be interesting to study the appropriate nonholonomic systems in arbitrary dimension $\R^m$, $m>3$ (e.g., see \cite{FK1995, FJ1, Jo1, Naranjo2019b, Jov2019}),
or the systems where the homogeneous balls $\mathbf B_i$ are replaced by the systems of the form (ball + gyroscope), which satisfy the Zhukovskii conditions (see \cite{Zhuk1893, DGJ}).

\section{Planar system - the three balls bearings problem}\label{sec6}

\subsection{Definition of the planar three balls bearing problem}
Consider the limit, when the radii of the spheres $\mathbf S_0$ and $\mathbf S$ both tend to infinity.
For simplicity, we consider the case $n=3$.
As a result, we obtain rolling without slipping of three homogeneous balls $\mathbf B_1,\mathbf B_2,\mathbf B_3$
of the radius $r$ and masses $m_1, m_2, m_3$ over the fixed plane $\Sigma_0$, together with the moving plane $\Sigma$
of the mass $m$ that is placed over the balls, such that there is no slipping between the balls and moving plane.
We will refer to the system as \emph{the planar three balls bearing problem}.
Note that all considerations of the Section can be easily adopted for the case of the
planar ball bearing with rolling of $n$ homogeneous balls.

Let $O_0$ be the fixed point of the plane $\Sigma_0$,
$O, O_1, O_2, O_3$ be the centers of mass of the plane $\Sigma$ and the balls $\mathbf B_1,\mathbf B_2,\mathbf B_3$ respectively. Let also
\[
O_0\vec{\mathbf e}^0_1,\vec{\mathbf e}^0_2,\vec{\mathbf e}^0_3, \qquad O\vec{\mathbf e}_1,\vec{\mathbf e}_2,\vec{\mathbf e}_3,
\qquad O_i\vec{\mathbf e}^i_1,\vec{\mathbf e}^i_2,\vec{\mathbf e}^i_3, ,
\]
be positively oriented reference frames rigidly attached to the fixed plane $\Sigma_0$, the moving plane $\Sigma$, and the ball $\mathbf B_i$ ($i=1,2,3$), respectively.
Here $\vec{\mathbf e}_3=\vec{\mathbf e}^0_3$ is the unit vector orthogonal to $\Sigma$ and $\Sigma_0$.

In the fixed reference frame, the positions of the points $O$, $O_1$, $O_2$, and $O_3$ are respectively given by
\[
O(x,y,2r), \qquad O_1(x_1,y_1,r), \qquad O_2(x_2,y_2,r), \qquad O_3(x_3,y_3,r).
\]

We denote by $\mathbf g\in SO(2)\subset SO(3)$  the rotation matrix that maps $O\vec{\mathbf e}_1,\vec{\mathbf e}_2,\vec{\mathbf e}_3$ to  $O_0\vec{\mathbf e}^0_1,\vec{\mathbf e}^0_2\vec{\mathbf e}_3^0$, and by $\mathbf g_i\in SO(3)$ the matrix that maps the moving frame $O_i\vec{\mathbf e}^i_1,\vec{\mathbf e}^i_2,\vec{\mathbf e}^i_3$
to the fixed frame $O_0\vec{\mathbf e}^0_1,\vec{\mathbf e}^0_2,\vec{\mathbf e}^0_3$, $i=1,2,3$. As above,
the skew-symmetric matrices
\[
\omega=\dot{\mathbf g}{\mathbf g}^{-1}, \qquad \omega_i=\dot{\mathbf g}_i \mathbf g_i^{-1},
\]
after the identification \eqref{izomorfizam}, correspond to the angular velocities $\vec{\omega}$ and  $\vec\omega_i$
of the plane $\Sigma$ and  the ball $\mathbf B_i$ relative to the fixed coordinate system. Note that
\[
\mathbf g=
\begin{pmatrix}
\cos\varphi & -\sin\varphi & 0 \\
\sin\varphi & \cos\varphi & 0 \\
0 &  0 & 1
\end{pmatrix}, \quad
\omega=
\begin{pmatrix}
0 &  - \dot\varphi & 0 \\
\dot\varphi & 0 & 0 \\
0 &  0 & 0
\end{pmatrix},
\quad \text{and} \quad \vec\omega=(0,0,\dot\varphi).
\]

Then the configuration space of the planar three balls bearing problem is
\[
Q=SO(3)\times SO(3)\times SO(3) \times \R^2\times SO(2)\times (\R^2)^3\{\mathbf g_1,\mathbf g_2,\mathbf g_3,x,y,\varphi,x_1,y_1,x_2,y_2,x_3,y_3\},
\]
while the kinetic energy is
\[
T=\frac12 Iv_\varphi^2+\frac12 m\langle \vec{v}_O,\vec{v}_O\rangle+\frac12\sum_{i=1}^3 I_i \langle \vec{\omega}_i,\vec{\omega}_i\rangle+\frac12 \sum_{i=1}^3 m_i \langle \vec{v}_{O_i},\vec{v}_{O_i}\rangle,
\]
where $\diag(I_i,I_i,I_i)$ is the inertia operator of the ball $\mathbf B_i$, $i=1,2,3$, $I$ is the moment of inertia of the plane $\Sigma$ for the $O\vec{\mathbf e}_3$-axis through the mass centre $O$.

\begin{pspicture}(12,7)
\pscircle[linecolor=black,fillstyle=solid, fillcolor=gray!20](3.5,3){1}
\psellipticarc[linestyle=dashed](3.5,3)(1,0.2){0}{180}
\psellipticarc(3.5,3)(1,0.2){180}{360}
\psdot[dotsize=2pt](3.5,3)\uput[0](3.5,3.2){$O_1$}
\psdot[dotsize=1.5pt](3.5,4)\uput[0](3.5,4.2){$A$}
\psdot[dotsize=1.5pt](3.5,2)\uput[0](3.5,1.8){$A_1$}

\pscircle[linecolor=black, fillstyle=solid, fillcolor=gray!20](7,3){1}
\psellipticarc[linestyle=dashed](7,3)(1,0.2){0}{180}
\psellipticarc(7,3)(1,0.2){180}{360}
\psdot[dotsize=2pt](7,3)\uput[0](7,3.2){$O_2$}
\psdot[dotsize=1.5pt](7,4)\uput[0](7,4.2){$B$}
\psdot[dotsize=1.5pt](7,2)\uput[0](7,1.8){$B_1$}

\pscircle[linecolor=black, fillstyle=solid, fillcolor=gray!20](5.5,5){0.8}
\psellipticarc[linestyle=dashed](5.5,5)(.8,0.2){0}{180}
\psellipticarc(5.5,5)(0.8,0.2){180}{360}
\psdot[dotsize=2pt](5.5,5)\uput[0](5.5,5.2){$O_3$}
\psdot[dotsize=1.5pt](5.5,5.8)\uput[0](5.5,6){$C$}
\psdot[dotsize=1.5pt](5.5,4.2)\uput[0](5.5,4){$C_1$}

\psline[linecolor=black,linewidth=0.01cm](1,1)(10,1)
\psline[linecolor=black, linewidth=0.01cm](1,1)(1.60,3.5)
\psline[linecolor=black, linewidth=0.01cm, linestyle=dashed](1.60,3.5)(2,5.2)
\psline[linecolor=black, linewidth=0.01cm, linestyle=dashed](2,5.2)(4.7,5.2)
\psline[linecolor=black, linewidth=0.01cm, linestyle=dashed](6.3,5.2)(9,5.2)
\psline[linecolor=black, linewidth=0.01cm, linestyle=dashed](9,5.2)(9.39,3.5)
\psline[linecolor=black, linewidth=0.01cm,](9.39,3.5)(10,1)

\psline[linecolor=black,linewidth=0.02cm](1,3.5)(10,3.5)
\psline[linecolor=black, linewidth=0.02cm](1,3.5)(2,6)
\psline[linecolor=black, linewidth=0.02cm](2,6)(9,6)
\psline[linecolor=black, linewidth=0.02cm](9,6)(10,3.5)
\uput[0](1.1,3.7){$\Sigma$}
\uput[0](2,0.5){{\sc Figure 2}. Planar three balls bearing problem}
\end{pspicture}

Let $A_1$ and $A$ be the points of contact of $\mathbf B_1$ with
fixed plane $\Sigma_0$ and with plane $\Sigma$, $B_1$  and $B$ be the contact points of $\mathbf B_2$, and $C_1$ and $C$ be the contact points of $\mathbf B_3$ with those planes. We have the following nonholonomic constraints written in the fixed reference frame $O_0\vec{\mathbf e}^0_1,\vec{\mathbf e}^0_2,\vec{\mathbf e}^0_3$:
\begin{equation}\label{noncon}
\begin{aligned}
\vec{v}_{O_1}&-r\vec{\omega}_1\times\vec{\gamma}=0,\ \vec{v}_{O_2}-r\vec{\omega}_2\times\vec{\gamma}=0,\ \vec{v}_{O_3}-r\vec{\omega}_3\times\vec{\gamma}=0,\\
\vec{v}_{O_1}&+r\vec{\omega}_1\times\vec{\gamma}=\vec{v}_O+\vec{\omega}\times\overrightarrow{OA},\\
\vec{v}_{O_2}&+r\vec{\omega}_2\times\vec{\gamma}=\vec{v}_O+\vec{\omega}\times\overrightarrow{OB},\\
\vec{v}_{O_3}&+r\vec{\omega}_3\times\vec{\gamma}=\vec{v}_O+\vec{\omega}\times\overrightarrow{OC},
\end{aligned}
\end{equation}
where $\vec\gamma=(0,0,1)$ is the unit vector orthogonal to the planes $\Sigma_0$ and $\Sigma$, i.e.,
\[
\vec\gamma=\vec{\mathbf e}_3^0=\vec{\mathbf e}_3.
\]

The first three vector constraints are obtained from the condition that the velocities of the contact points $A_1, B_1, C_1$ with
the fixed plane $\Sigma_0$ are zero. The remaining ones follow
from the condition that there is no sliding between the balls and the plane $\Sigma$. This means that the velocities of points $A$, $B$ and $C$ are the same
as velocities of the corresponding points at the plane $\Sigma$.

The configuration space $Q$ is
18-dimensional and there are twelve independent nonholonomic constraints among \eqref{noncon}. Hence 
the vector subspaces of the admissible
velocities $\mathcal D_q\subset T_q Q,$ $q\in Q$, are six--dimensional.

Note that the dimensions of the configuration space $Q$ and the constraints manifold
$\mathcal D$ in the problem of spherical ball bearing for $n=3$ and the planar three balls bearing problem coincide.
Now, the additional one-side constraints read:
\begin{equation}\label{jednostrane}
\begin{aligned}
&\vert\overrightarrow{O_1O_2}\vert=\sqrt{(x_2-x_1)^2+(y_2-y_1)^2}  \ge 2r, \\
&\vert\overrightarrow{O_2O_3}\vert=\sqrt{(x_3-x_2)^2+(y_3-y_2)^2} \ge 2r, \\
&\vert\overrightarrow{O_3O_1}\vert=\sqrt{(x_1-x_3)^2+(y_1-y_3)^2}\ge 2r.
\end{aligned}
\end{equation}

\subsection{The equations of motion}
As in the case of spherical ball bearing problem, we will derive the equations of  motion in the planar case in the vector form. The corresponding reaction forces
will be expressed in terms of Lagrange multipliers $\lambda_1,\dots,\lambda_{12}$
that correspond to twelve independent nonholonomic costraints.

The equations of motion of the panar three balls bearing system relative to the fixed coordinate system are:
\begin{equation}\label{jednacine}
\begin{aligned}
m_1 \dot{\vec{v}}_{O_1}&=(\lambda_1,\lambda_2,0)+(\lambda_7,\lambda_8,0),\\
m_2\dot{\vec{v}}_{O_2}&=(\lambda_3,\lambda_4,0)+(\lambda_9,\lambda_{10},0),\\
m_3\dot{\vec{v}}_{O_3}&=(\lambda_5,\lambda_6,0)+(\lambda_{11},\lambda_{12},0),\\
m\dot{\vec{v}}_O&=-(\lambda_7,\lambda_8,0)-(\lambda_9,\lambda_{10},0)-(\lambda_{11},\lambda_{12},0),
\end{aligned}
\end{equation}
and
\begin{equation}\label{jednacine*}
\begin{aligned}
{I}_1\dot{\vec{\omega}}_1&=-r\vec{\gamma}\times((\lambda_1,\lambda_2,0)-(\lambda_7,\lambda_8,0)),\\
{I}_2\dot{\vec{\omega}}_2&=-r\vec{\gamma}\times((\lambda_3,\lambda_4,0)-(\lambda_9,\lambda_{10},0),\\
{I}_3\dot{\vec{\omega}}_3&=-r\vec{\gamma}\times((\lambda_5,\lambda_6,0)-(\lambda_{11},\lambda_{12},0)), \\
I\dot{\vec{\omega}}&=-\overrightarrow{OA}\times(\lambda_7,\lambda_8,0)-\overrightarrow{OB}\times(\lambda_9,\lambda_{10},0)-
\overrightarrow{OC}\times(\lambda_{11},\lambda_{12},0).
\end{aligned}
\end{equation}

 By differentiating the first three constraints from \eqref{noncon}, and using the first three equations of motion in \eqref{jednacine} and \eqref{jednacine*},
we get:
\begin{equation}\label{mnozioci1}
\begin{aligned}
(\lambda_1,\lambda_2, 0)&=\frac{m_1r^2-I_1}{m_1r^2+I_1}(\lambda_7,\lambda_8,0)\\
(\lambda_3,\lambda_4, 0)&=\frac{m_2r^2-I_2}{m_2r^2+I_2}(\lambda_9,\lambda_{10},0),\\
(\lambda_5,\lambda_6, 0)&=\frac{m_3r^2-I_3}{m_3r^2+I_3}(\lambda_{11},\lambda_{12},0).
\end{aligned}
\end{equation}

Therefore, the equations \eqref{jednacine}, \eqref{jednacine*} can be written as
\begin{equation}\label{jed1}
\begin{aligned}
\dot{\vec{v}}_{O_1}&=\frac{2r^2}{m_1r^2+I_1}\vec{\mathbf F}_1,\\
\dot{\vec{v}}_{O_2}&=\frac{2r^2}{m_2r^2+I_2}\vec{\mathbf F}_2,\\
\dot{\vec{v}}_{O_3}&=\frac{2r^2}{m_3r^2+I_3}\vec{\mathbf F}_3,\\
m\dot{\vec{v}}_O&=-\vec{\mathbf F}_1-\vec{\mathbf F}_2-\vec{\mathbf F}_3,\\
\end{aligned}
\end{equation}
and
\begin{equation}\label{jed2}
\begin{aligned}
\dot{\vec{\omega}}_1&=\frac{2r}{m_1r^2+I_1}\vec{\gamma}\times\vec{\mathbf F}_1,\\
\dot{\vec{\omega}}_2&=\frac{2r}{m_2r^2+I_2}\vec{\gamma}\times\vec{\mathbf F}_2,\\
\dot{\vec{\omega}}_3&=\frac{2r}{m_3r^2+I_3}\vec{\gamma}\times\vec{\mathbf F}_3,\\
I\dot{\vec{\omega}}&=-\overrightarrow{OA}\times\vec{\mathbf F}_1-\overrightarrow{OB}\times\vec{\mathbf F}_2-
\overrightarrow{OC}\times\vec{\mathbf F}_3,
\end{aligned}
\end{equation}
where
\[
\vec{\mathbf F}_1=(\lambda_7,\lambda_8,0), \qquad \vec{\mathbf F}_2=(\lambda_9,\lambda_{10},0), \qquad \vec{\mathbf F}_3=(\lambda_{11},\lambda_{12},0).
\]

By differentiating the remaining constraints, we get the following linear system of six equations in the Lagrange multipliers $\lambda_7,\dots,\lambda_{12}$:
\begin{equation}\label{mnozioci2}
\begin{aligned}
\frac{4Ir^2}{m_1r^2+I_1}\vec{\mathbf F}_1=&\big(\overrightarrow{OA}\wedge \vec{\mathbf F}_1+\overrightarrow{OB}\wedge \vec{\mathbf F}_2+\overrightarrow{OC}\wedge \vec{\mathbf F}_3\big)\overrightarrow{OA}\\
&+\vec{\omega}\times(\vec{v}_{O_1}-\vec{v}_{O})-\frac{I}{m}\big( \vec{\mathbf F}_1+\vec{\mathbf F}_2+\vec{\mathbf F}_3\big),\\
\frac{4Ir^2}{m_2r^2+I_2}\vec{\mathbf F}_2=
&\big(\overrightarrow{OA}\wedge \vec{\mathbf F}_1+\overrightarrow{OB}\wedge \vec{\mathbf F}_2+\overrightarrow{OC}\wedge \vec{\mathbf F}_3\big)\overrightarrow{OB}\\
&+\vec{\omega}\times(\vec{v}_{O_2}-\vec{v}_{O})-\frac{I}{m}\big( \vec{\mathbf F}_1+\vec{\mathbf F}_2+\vec{\mathbf F}_3\big),\\
\frac{4Ir^2}{m_3r^2+I_3}\vec{\mathbf F}_3=
&\big(\overrightarrow{OA}\wedge \vec{\mathbf F}_1+\overrightarrow{OB}\wedge \vec{\mathbf F}_2+\overrightarrow{OC}\wedge \vec{\mathbf F}_3\big)\overrightarrow{OC}\\
&+\vec{\omega}\times(\vec{v}_{O_3}-\vec{v}_{O})-\frac{I}{m}\big( \vec{\mathbf F}_1+\vec{\mathbf F}_2+\vec{\mathbf F}_3\big).
\end{aligned}
\end{equation}
One can easily see that the system \eqref{mnozioci2} determines
the Lagrange multipliers $\lambda_7,\dots,\lambda_{12}$ uniquely, and, at the same time, uniquely determines  $\vec{\mathbf F}_1, \vec{\mathbf F}_2, \vec{\mathbf F}_3$.

We have the following analogue of Propositions \ref{prva} and \ref{go}

\begin{prop}
The moving triangles $\triangle O_1O_2O_3(t)$ and $\triangle ABC(t)$ are congruent to the triangle formed by the centers of the balls at the initial condition.
\end{prop}

\begin{proof}
From the constraints \eqref{noncon} we get
\begin{equation*}
\begin{aligned}
&2\frac{d}{dt}\overrightarrow{AB}=\vec{\omega}\times\overrightarrow{AB},\\
&2\frac{d}{dt}\overrightarrow{BC}=\vec{\omega}\times\overrightarrow{BC},\\
&2\frac{d}{dt}\overrightarrow{CA}=\vec{\omega}\times\overrightarrow{CA}.
\end{aligned}
\end{equation*}

Therefore,
\[
\frac{d}{dt}\langle \overrightarrow{AB},\overrightarrow{AB}\rangle=\langle \overrightarrow{AB},\vec{\omega}\times\overrightarrow{AB}\rangle=0.
\]
Similarly, we have $\langle \overrightarrow{BC},\overrightarrow{BC}\rangle=const$, $\langle \overrightarrow{CA},\overrightarrow{CA}\rangle=const$.
\end{proof}

Thus, as in the spherical case, if the initial condition is within the interior of the region \eqref{jednostrane}, the system remains
within the interior of the region \eqref{jednostrane} along the motion.

Also, from \eqref{jed2} and $\dot\gamma=0$ we get:

\begin{prop} The projections of the angular velocities $\vec{\omega}_i$ to $\vec{\gamma}$ are conserved along the motion:
\[
\omega_{i3}=\langle \vec{\omega}_i,\vec\gamma\rangle=c_i, \qquad i=1,2,3.
\]
\end{prop}

\subsection{Reduction}

Set $v_\varphi=\dot\varphi$, $v_x=\dot x$, $v_y=\dot y$.
By using the constraints \eqref{noncon} we can obtain a closed system of the equations of motion on the space
\[
\mathcal P= T\R^2\times TSO(2)\times (\R^2)^3\{v_x,v_y,v_\varphi, x,y, \varphi, x_1,y_1,x_2,y_2,x_3,y_3\}, \qquad \dim\mathcal P=12.
\]

Note that we have a diffeomorphism
\[
\mathcal D/SO(3)\times SO(3)\times SO(3) \cong \mathcal P\times\R^3\{\omega_{13},\omega_{23},\omega_{33}\},
\]
where the $SO(3)\times SO(3)\times SO(3)$--action on $\mathcal D\subset TQ$, as in the spherical case, is given by the right trivialisation of the tangent bundle
of the Lie group $SO(3)\times SO(3)\times SO(3)$.

It is interesting that, contrary to the spherical case, the equations on $\mathcal P$ do not depend on the integrals $c_i$.
By using the constraints \eqref{noncon}, the kinetic energy for $c_1=c_2=c_3=0$ on $\mathcal P$ takes the form:
\begin{align*}
T= \frac12 Iv_\varphi^2+ \frac12 m\big(v_x^2+v_y^2\big)&+\frac12 (\frac{I_1+r^2m_1}{4r^2}) \langle \vec{v}_O+\vec{\omega}\times\overrightarrow{OA},\vec{v}_O+\vec{\omega}\times\overrightarrow{OA}\rangle\\
&+\frac12 (\frac{I_2+r^2m_2}{4r^2}) \langle \vec{v}_O+\vec{\omega}\times\overrightarrow{OB},\vec{v}_O+\vec{\omega}\times\overrightarrow{OB}\rangle\\
&+\frac12 (\frac{I_3+r^2m_3}{4r^2}) \langle \vec{v}_O+\vec{\omega}\times\overrightarrow{OC},\vec{v}_O+\vec{\omega}\times\overrightarrow{OC}\rangle.
\end{align*}

Let us denote
\begin{equation}\label{oznake}
\begin{aligned}
\vec N=&\delta_1\overrightarrow{OA}+\delta_2\overrightarrow{OB}+\delta_3\overrightarrow{OC}, \\
     M=&\delta_1\langle\overrightarrow{OA},\overrightarrow{OA}\rangle+
\delta_2\langle\overrightarrow{OB},\overrightarrow{OB}\rangle+\delta_3\langle\overrightarrow{OC},\overrightarrow{OC}\rangle,\\
\delta_1=&\frac{m_1r^2+I_1}{4r^2}, \quad  \delta_1=\frac{m_2r^2+I_2}{4r^2}, \quad \delta_3=\frac{m_3r^2+I_3}{4r^2}, \quad \delta=\delta_1+\delta_2+\delta_3.
\end{aligned}
\end{equation}

Note that
$\vec N(t)$ determines the trajectory of the mass centre $S(t)$ of the moving triangle $\triangle ABC(t)$ with masses $\delta_1, \delta_2, \delta_3$
placed at the vertices $A, B, C$: $\overrightarrow{OS}=\frac{1}{\delta}\vec N$.
Also, by definition, $\vec N$ and $M$ satisfy the inequality
$$
\delta M \ge \langle \vec N,\vec N\rangle=N_1^2+N_2^2.
$$
The equality would imply that the points $A, B, C$ coincide. Therefore, in the region of admissible motions \eqref{jednostrane} we have
\[
\delta M > N_1^2+N_2^2.
\]

With the above notation, the formula for the kinetic energy simplifies to
\[
T=\frac12 (I+M) v_\varphi^2+ \frac12 (m+\delta)(v_x^2+v_y^2)+v_\varphi(N_1 v_y-N_2 v_x).
\]

As in the problem of spherical ball bearing, we will derive the equations of motion in the planar case without calculating the explicit formulae for $\vec{\mathbf F}_1, \vec{\mathbf F}_2, \vec{\mathbf F}_3$.

\begin{thm}\label{Glavna2}
The equations of motion of the planar three balls bearing problem on $\mathcal P$ are given by
\begin{equation}\label{kinematicke}
\begin{aligned}
&\dot \varphi=v_\varphi, \qquad \dot x=v_x, \qquad \dot y=v_y, \\
&2(\dot x_1,\dot y_1)=(v_x,v_y)+(-v_\varphi(y_1-y),v_\varphi(x_1-x)),\\
&2(\dot x_2,\dot y_2)=(v_x,v_y)+(-v_\varphi(y_2-y),v_\varphi(x_2-x)),\\
&2(\dot x_3,\dot y_3)=(v_x,v_y)+(-v_\varphi(y_3-y),v_\varphi(x_3-x)),
\end{aligned}
\end{equation}
and
\begin{equation}\label{trikugle}
\begin{aligned}
(m+\delta)\dot v_x=&\frac12 N_1 v_\varphi^2-\frac\delta2 v_\varphi v_y+N_2\dot{v}_\varphi,\\
(m+\delta)\dot v_y=&\frac12 N_2 v_\varphi^2+\frac\delta2 v_\varphi v_x-N_1\dot{v}_\varphi,\\
\big(I+M\big)\dot{v}_\varphi=&\frac12 v_\varphi(N_1 v_x+N_2 v_y)+N_2 \dot v_x-N_1\dot v_y,
\end{aligned}
\end{equation}
where $\vec N$, $M$, 
${\delta}$ are given by \eqref{oznake}.
\end{thm}

An explicit form of the equations \eqref{trikugle} is given below in \eqref{REDtrikugle}.

\begin{proof}
The kinematic equations \eqref{kinematicke} follow directly from the constraints \eqref{noncon}.

From the last equations in \eqref{jed1} and \eqref{jed2}, we have:
\begin{equation}
(m\dot v_x,m\dot v_y, I\dot v_\varphi)=-\vec{\mathbf F}_1-\vec{\mathbf F}_2-\vec{\mathbf F}_3
-\overrightarrow{OA}\times\vec{\mathbf F}_1-\overrightarrow{OB}\times\vec{\mathbf F}_2-
\overrightarrow{OC}\times\vec{\mathbf F}_3,
\end{equation}
where $\vec{\mathbf F}_1, \vec{\mathbf F}_2, \vec{\mathbf F}_3$, from
\eqref{mnozioci2} and \eqref{kinematicke}, are written in terms of variables on $\mathcal P$.

Then, from \eqref{jed1} and  \eqref{noncon}, we get
\begin{equation*}
\begin{aligned}
&\vec{\mathbf F}_1=2\delta_1 \dot{\vec v}_{O_1}=\delta_1\frac{d}{dt}\big(\vec{v}_O+\vec{\omega}\times\overrightarrow{OA}\big),\\
&\vec{\mathbf F}_2=2\delta_2 \dot{\vec v}_{O_2}=\delta_2\frac{d}{dt}\big(\vec{v}_O+\vec{\omega}\times\overrightarrow{OB}\big), \\
&\vec{\mathbf F}_3=2\delta_1 \dot{\vec v}_{O_3}=\delta_3\frac{d}{dt}\big(\vec{v}_O+\vec{\omega}\times\overrightarrow{OC}\big).
\end{aligned}
\end{equation*}

Thus, the last equation in \eqref{jed1} can be rewritten as
\begin{equation}\label{dotVo}
\frac{d}{dt}\Big((m+\delta)\vec{v}_O+ \vec{\omega}\times \vec N \Big)=0, \qquad \delta=\delta_1+\delta_2+\delta_3.
\end{equation}

On the other hand, we have
\begin{align*}
\overrightarrow{OA}\times\vec{\mathbf F}_1=&\delta_1\overrightarrow{OA}\times\frac{d}{dt}\big(\vec{v}_O+\vec{\omega}\times\overrightarrow{OA}\big)\\
=&\frac{d}{dt}\big(\delta_1\overrightarrow{OA}\times \big(\vec{v}_O+\vec{\omega}\times\overrightarrow{OA}\big)\big)-\delta_1\big({\vec v}_{O_1}-{\vec v}_{O}\big)
\times \big(\vec{v}_O+\vec{\omega}\times\overrightarrow{OA}\big)\\
=&\frac{d}{dt}\big(\delta_1\overrightarrow{OA}\times \big(\vec{v}_O+\vec{\omega}\times\overrightarrow{OA}\big)\big)
-\frac{\delta_1}2\big(\vec{\omega}\times\overrightarrow{OA}-{\vec v}_{O}\big)
\times \big(\vec{v}_O+\vec{\omega}\times\overrightarrow{OA}\big)\\
=&\frac{d}{dt}\big(\delta_1\overrightarrow{OA}\times \vec{v}_O + \delta_1\overrightarrow{OA}\times\big(\vec{\omega}\times\overrightarrow{OA}\big)\big)
+{\vec v}_{O}\times \big(\vec{\omega}\times \delta_1\overrightarrow{OA}\big)\\
=&\frac{d}{dt}\big(\delta_1\overrightarrow{OA}\times \vec{v}_O + \delta_1\vec{\omega} \langle\overrightarrow{OA},\overrightarrow{OA}\rangle\big)
+\vec{\omega}\langle {\vec v}_{O}  ,  \delta_1 \overrightarrow{OA}\rangle.
\end{align*}
Similar equations hold for $\overrightarrow{OB}\times\vec{\mathbf F}_2$ and
$\overrightarrow{OC}\times\vec{\mathbf F}_3$. Therefore, the last equation in \eqref{jed2} takes the form
\begin{equation}\label{dotW}
\frac{d}{dt}\Big(I{\vec{\omega}}+\vec N\times \vec{v}_O+M\vec{\omega}\Big)+\vec{\omega}\langle {\vec v}_{O}  , \vec N \rangle=0.
\end{equation}

The time derivatives of  $M$ and $\vec N$ along the motion are given by
\begin{equation*}
\begin{aligned}
\dot M=&2\delta_1\langle\overrightarrow{OA},\vec{v}_{O_1}-\vec{v}_O\rangle+
2\delta_2\langle\overrightarrow{OB},\vec{v}_{O_2}-\vec{v}_O\rangle+2\delta_3\langle\overrightarrow{OC},\vec{v}_{O_3}-\vec{v}_O\rangle\\
=&\delta_1\langle\overrightarrow{OA},\vec{\omega}\times\overrightarrow{OA}-{\vec v}_{O}\rangle+
\delta_2\langle\overrightarrow{OB},\vec{\omega}\times\overrightarrow{OB}-{\vec v}_{O}\rangle+\delta_3\langle\overrightarrow{OC},\vec{\omega}\times\overrightarrow{OC}-{\vec v}_{O}\rangle\\
=& - \langle {\vec v}_{O}  , \vec N \rangle,\\
\dot{\vec N}=&
\frac{\delta_1}2\big(\vec{\omega}\times\overrightarrow{OA}-{\vec v}_{O}\big)+
\frac{\delta_2}2\big(\vec{\omega}\times\overrightarrow{OB}-{\vec v}_{O}\big)+
\frac{\delta_3}2\big(\vec{\omega}\times\overrightarrow{OC}-{\vec v}_{O}\big)\\
=&\frac12\big(\vec\omega\times \vec N-\delta\vec{v}_O\big).
\end{aligned}
\end{equation*}

Finally,  the equations \eqref{dotVo} and \eqref{dotW} can be written as:
\begin{equation*}
\begin{aligned}
(m+\delta)\dot{\vec{v}}_O+
\big(I+M\big)\dot{\vec{\omega}}=&\frac{d}{dt}\big(\vec N\times (\vec{\omega}-\vec{v}_O)\big)\\
=&\frac12\langle \vec\omega,\vec\omega\rangle \vec N+\frac12 \langle {\vec v}_{O}  , \vec N \rangle\vec\omega+\frac{\delta}2 \vec\omega \times\vec{v}_O+
\vec N \times \big(\dot{\vec\omega}-\dot{\vec{v}}_O\big),
\end{aligned}
\end{equation*}
which proves \eqref{trikugle}.
\end{proof}

\subsection{Invariant measure}
It is clear that we can pass from $\mathcal P$ to the space
\begin{equation}\label{eq:Q}
\mathcal Q=\{(v_x,v_y,v_\varphi,N_1,N_2,M)\in \R^6\,\vert\, \delta M > N_1^2+N_2^2\},
\end{equation}
with the induced system described by the equations \eqref{trikugle} and
\[
\dot{\vec N}=\frac12\big(\vec\omega\times \vec N-\delta\vec{v}_O\big), \qquad \dot M=- \langle {\vec v}_{O}  , \vec N \rangle.
\]

If we introduce
\begin{align*}
&\mathbf v=(v_x,v_y,v_\varphi), \qquad
 \mathbf n=(N_1,N_2,M),\\
&\mathbf m=\frac12(N_1 v_\varphi^2-\delta v_\varphi v_y,N_2 v_\varphi^2+\delta v_\varphi v_x, v_\varphi(N_1 v_x+N_2 v_y)),\\
&\mathbb I=
\begin{pmatrix}
m+\delta & 0            & -N_2 \\
0        & m+\delta     &  N_1  \\
-N_2     & N_1          &  I+M
\end{pmatrix},\qquad
\mathbb J=-\frac 12
\begin{pmatrix}
\delta & 0 &      N_2\\
0     & \delta & -N_1\\
  2N_1 & 2N_2 & 0
\end{pmatrix},
\end{align*}
then the reduced equations of motion on $\mathcal Q$ \eqref{eq:Q} become
\begin{equation}\label{REDtrikugle}
\dot{\mathbf v}=\mathbb I^{-1}\mathbf m, \qquad \dot{\mathbf n}=\mathbb J \mathbf v.
\end{equation}

\begin{remark}
Since
\[
\det (\mathbb I)=(m+\delta)\big((m+\delta)I + m M+ (\delta M-(N_1^2+N_2^2) \big)> 0\vert_\mathcal Q,
\]
the matrix $\mathbb I$ is invertible on $\mathcal Q$ and we have
\[
\mathbb I^{-1}=\frac{1}{\det(\mathbb I)}
\begin{pmatrix}
(m+\delta)(I+M)-N_1^2 & -N_1 N_2            & (m+\delta)N_2 \\
-N_1N_2        &  (m+\delta)(I+M)-N_2^2  &-(m+\delta)N_1  \\
(m+\delta)N_2     & -(m+\delta)N_1          &  (m+\delta)^2
\end{pmatrix}.
\]
\end{remark}

\begin{thm}\label{mera2}
The equations \eqref{REDtrikugle} have the following first integrals
\begin{equation}\label{integraliRED}
\begin{aligned}
& f_1=(m+\delta) v_x- v_\varphi N_2, \\
& f_2=(m+\delta) v_y + v_\varphi N_1, \\
& f_3=\delta M-(N_1^2+N_2^2),\\
& f_4=T=\frac12 (I+M) v_\varphi^2+ \frac12 (m+\delta)(v_x^2+v_y^2)+v_\varphi(N_1 v_y-N_2 v_x).
\end{aligned}
\end{equation}
and they possess an invariant measure
\[
\sqrt{\det(\mathbb I)}\, dv_x \wedge dv_y \wedge dv_\varphi \wedge dN_1 \wedge dN_2 \wedge dM.
\]
The system \eqref{REDtrikugle} can be solved by quadratures.
\end{thm}

\begin{proof}
It is clear that the functions \eqref{integraliRED} are the first integrals of the system \eqref{REDtrikugle}.
Note that $f_3>0$ on $\mathcal Q$.
Next, at the invariant level set
\[
\mathcal Q_d\colon \qquad f_1=d_1,\qquad f_2=d_2, \qquad f_3=d_3,
\]
we have
\begin{align*}
&v_x=\frac{v_\varphi N_2+d_1}{m+\delta}, \qquad
 v_y=\frac{-v_\varphi N_1+d_2}{m+\delta},\qquad M=\frac{1}{\delta}(N_1^2+N_2^2)+\frac{d_3}{\delta},\\
&\det (\mathbb I)=(m+\delta)\big((m+\delta)I + \frac{m}{\delta}(N_1^2+N_2^2)+ \frac{m d_3}{\delta}+d_3 \big).
\end{align*}
We obtain a closed system in the space $\R^3\{v_\varphi,N_1,N_2\}$ given by
\begin{equation}\label{trikugle*}
\begin{aligned}
\dot{v}_\varphi=&\frac{m v_\varphi(N_1 {d_1}+N_2 {d_2})}{2\det(\mathbb I)},\\
\dot N_1
=&-\frac{m+2\delta}{2(m+\delta)}N_2 v_{\varphi}-\frac{\delta d_1}{2(m+\delta)},\\
\dot N_2
=&\frac{m+2\delta}{2(m+\delta)}N_1 v_{\varphi}-\frac{\delta d_2}{2(m+\delta)}.
\end{aligned}
\end{equation}

Using the similar arguments as in the proof Theorem \ref{mera},
it is sufficient to prove that $\mu\vert_{\mathcal Q_d}$ is the density of an invariant measure of the
reduced system \eqref{trikugle*}.

Let $X=(\dot v_\varphi,\dot N_1,\dot N_2)$. Then
\[
\dv(X)=\frac{m(N_1d_1+N_2d_2)}{2\det(\mathbb I)}.
\]
On the other hand
\[
\frac{d}{dt}\det(\mathbb I)=-m(N_1d_1+N_2d_2).
\]

Therefore, the function $\mu=\sqrt{\det(\mathbb I)}$ satisfies the equation
\[
\dot \mu+\mu\dv(X)=0,
\]
and the system \eqref{trikugle*} preserves the measure $\mu\, dv_\varphi \wedge dN_1 \wedge dN_2$. Integrability in quadratures follows according the Euler-Jacobi theorem \cite{AKN}.
\end{proof}

\begin{remark}
By setting $d_1=d_2=0$,  we get that
\[
v_\varphi=const \qquad \text{and} \qquad N_1^2+N_2^2=const.
\]
Thus, the equations
\eqref{trikugle*} can be solved in terms of trigonometric functions.
\end{remark}

\subsection*{Acknowledgements}
We are very grateful to the referees for valuable remarks that helped us to improve the exposition.
The research which led to this paper was initiated during the GDIS conference in Summer 2018, when all three authors visited Moscow Institute of Physics and Technology, kindly invited and hosted by Professor Alexey V. Borisov and his team.
This research has been supported by the Project no. 7744592 MEGIC "Integrability and Extremal Problems in Mechanics, Geometry and Combinatorics" of the Science Fund of Serbia,
Mathematical Institute of the Serbian Academy of Sciences and Arts and the Ministry for Education, Science, and Technological Development of Serbia, and the Simons Foundation grant no. 854861.

\end{document}